\documentclass[letterpaper, 10 pt, conference]{ieeeconf}  

\IEEEoverridecommandlockouts                              

\usepackage{graphics}
\usepackage{amsmath}

\usepackage{amsthm}
\usepackage{comment}
\usepackage{amssymb}
\usepackage{mathtools}
\usepackage{bbm}
\usepackage{dsfont}
\usepackage{color}
\usepackage{xcolor}
\definecolor{myblue}{rgb}{0 0.4470 0.7410}
\definecolor{mypurple}{rgb}{0.4940 0.1840 0.5560}
\definecolor{mygreen}{rgb}{0.4660 0.6740 0.1880}
\usepackage{cleveref}
\usepackage[noend]{algpseudocode}
\usepackage[ruled, vlined]{algorithm2e}
\newtheorem{lem}{Lemma}

\newtheorem{assumption}{Assumption}

\newtheorem{prop}{Proposition}

\newtheorem{definition}{Definition}
\newtheorem{remark}{Remark}
\theoremstyle{remark}

\usepackage{pgfplots}
\pgfplotsset{compat=newest}
\pgfplotsset{plot coordinates/math parser=false}
\newlength\figureheight
\newlength\figurewidth

\newcommand{\goutam}[1]{\textcolor{blue}{[Goutam: #1]}}

\newcommand{\R}{\mathbb{R}}

\title{\LARGE \bf
Bound-Optimized Task Choice \\ for Path Integral Control
}

\author{Rylie Anderson\textsuperscript{1}, Goutam Das\textsuperscript{2}, Takashi Tanaka\textsuperscript{3} \thanks{
This work is supported by DARPA COMPASS program grant HR0011-25-3-0210 and AFOSR DSCT program grant FA9550-25-1-0347. All authors are associated with the Networked Control Systems lab at Purdue University. Emails: \textsuperscript{1} ande1946@purdue.edu, \textsuperscript{2} das347@purdue.edu, and \textsuperscript{3} tanaka16@purdue.edu.
}}

\begin{document}

\maketitle
\thispagestyle{empty}
\pagestyle{empty}

\begin{abstract}

Path Integral (PI) control is a powerful sampling-based method for stochastic optimal control, but it requires a restrictive coupling between the noise covariance and the control cost matrix that is rarely satisfied in practice, particularly in aerospace and cyber-physical systems. We propose Bound-Optimized Task Choice (BOTC), a framework that optimizes over the entire space of valid approximations, termed tasks, satisfying the PI coupling constraint. We prove that every task provides an upper bound on the true cost-to-go and that BOTC minimizes this bound. We derive a change-of-measure formulation that enables evaluation of all candidate tasks from a single set of Monte Carlo samples, eliminating the need to resample for each candidate task. The resulting optimization is parameterized by a positive semi-definite matrix. Furthermore, we propose a novel Normal-Inverse-Wishart distribution-based importance sampling scheme to improve global optimization. We validate BOTC on a finite-horizon stochastic linear-quadratic regulator problem, demonstrating that it tracks the constrained optimum.
\end{abstract}

\section{Introduction}\label{sec: introduction}

\noindent Traditionally, non-linear, stochastic optimal control requires solving the non-linear Hamilton-Jacobi-Bellman (HJB) partial differential equation (PDE). However, the HJB is intractable for practical high-dimensional systems. PI control is an alternative stochastic optimal control framework based on Monte Carlo trajectory sampling \cite{kappen2005path}, \cite{kappen2007path}. The log-transformed HJB equation is linearized if a specific linear relationship between control cost and covariance holds:
\begin{equation}\label{pi_constraint}
    \Sigma_{(f)} = \lambda B_{(f)} R^{-1} B_{(f)}^T
\end{equation}
where $\Sigma_{(f)}$ is noise covariance, $B_{(f)}$ is the control effect matrix, and $R$ is quadratic control cost. The $(f)$ subscripts denote the control and covariance matrices acting on the full state space rather than the directly actuated subspace. Applying the Feynman-Kac Lemma \cite{oksendal2003stochastic}, the linearized HJB equation may be converted into an expectation over uncontrolled trajectory cost, which can be evaluated with Monte Carlo methods.

For applicable problems, PI control efficiently solves the HJB equation while avoiding its most common pitfalls: PI control is solved entirely forward in time and does not require explicit calculation of the cost gradient. 

PI control has since been further refined and successfully applied to numerous challenging control problems. \cite{todorov2007linearly}, \cite{todorov2009efficient} extended the framework to analogous discrete MDPs. \cite{williams2017model} proposed the highly practical MPPI algorithm, which introduced an iterative update law. \cite{williams2016aggressive} demonstrated the effectiveness of MPPI in the non-linear, noisy environment of aggressive driving. \cite{abraham_2020} further extended MPPI to encompass uncertainty in system dynamics and implemented PI control on a wide variety of robotics tasks. Recently, \cite{minarik2024modelpredictivepathintegral} demonstrated the feasibility of MPPI in a quadcopter navigation task.

Although the typical interpretation of \eqref{pi_constraint} is as an intuitive connection between state noise and control authority \cite{theodorou2011thesis}, the constraint has broader implications. In many interesting applications, including aerospace and cyber-physical systems, controls are designed primarily to navigate deterministic dynamics rather than simply reject perturbations. In such cases, as noted in \cite{williams2017model}, the constraint shifts the relative balance between actuator penalties in a manner with no physical meaning, impacting policy synthesis. \cite{satoh2017iterative} proposed an iterative solution to the same class of problems without assuming \eqref{pi_constraint}, but the proposed algorithm is both substantially more computationally intensive and requires sufficiently differentiable cost functions.

In practice, to meet the requirements of \eqref{pi_constraint}, control cost, covariance, or both must be modified. Two simple but informative approaches are to either: allow covariance to imply control cost; or allow control cost to imply covariance. In either case, the implied term must be scaled by $\lambda$ be to positive semi-definitely greater than its original. Both approaches distort the true control problem. The first approach overestimates control cost, favoring under-actuated strategies. The second disproportionately penalizes noise-sensitive trajectories, producing excessively risk-averse control strategies. Both options introduce a systematic bias into policy synthesis, yet both approaches actually solve a harder control problem, giving an upper bound on the true expected cost. Handling of this trade-off in the literature has varied. \cite{abraham_2020} employed the first approach, disregarding control cost entirely. \cite{williams2016aggressive} struck a balance between these extremes, but that balance requires manual tuning through trial-and-error.

Many practical control problems do not naturally satisfy the exact equality condition of \eqref{pi_constraint} and must be approximated to a harder problem that does. The two possible methods of approximation listed above are simply the two extremes of a continuous space of valid approximations. No existing work optimizes over the full space of valid problem approximations. In this work, we address this gap by introducing Bound-Optimized Task Choice (BOTC).

In this paper, we define and examine the space of all valid problem approximations and find the subspace of potentially optimal problem approximations. We present a method of estimating the expected cost as a function of a decision variable parameterizing this entire space without resampling. We explore a novel importance sampling scheme based on the Normal-Inverse-Wishart distribution. Finally, we validate BOTC against a simple LQR system.

The remainder of this paper is organized as follows. Section \ref{sec:path_integral_control} reviews PI control. Section \ref{sec:problem_approximation} defines the space of valid problem approximations. Section \ref{sec:task_choice} optimizes expected cost over the space of all valid approximations. Section \ref{sec:simulations} presents and analyzes the results of a simple LQR experiment.

\section{Path Integral Control}\label{sec:path_integral_control}
\noindent In this section, we review path integral control as a method of stochastic optimal control applicable to non-linear systems, affine in control, with quadratic control cost. Let $x \in \mathbb{R}^{n}$ be the state of the system. Let $u \in \mathbb{R}^{m}$ be the control signal. Let $B_{(f)} \in \mathbb{R}^{n \times m}$ be the control effect matrix, where $m$ is the number of control channels and $n$ is the number of state dimensions. Finally, let $d\xi \in \mathbb{R}^{n}$ be a Brownian disturbance. Therefore, the system may be modeled as:
\begin{equation}\label{system_def}
    dx = f(x, t)dt + B_{(f)} u dt + d \xi
\end{equation}
where $\mathbb{E}[d \xi_i d \xi_j ] = \Sigma_{(f), ij}$, so $\Sigma_{(f)}$ is the covariance of state noise. 

\begin{assumption}\label{assume:constant_sigmab}
    Both $\Sigma_{(f)}$ and $B_{(f)}$ are constant.
\end{assumption}

When constrained by \eqref{pi_constraint}, the range of $B_{(f)}$ and the positive support of the noise must share the same state dimensions. Then, $\Sigma_{(f)}$ and $B_{(f)}$ may be partitioned into zero and non-zero components:
\begin{equation}\label{bs_partition}
    B_{(f)} = \begin{bmatrix}
        0 \\
        B
    \end{bmatrix}, \qquad
    \Sigma_{(f)} = \begin{bmatrix}
        0 & 0 \\
        0 & \Sigma
    \end{bmatrix}
\end{equation}
where $B \in \mathbb{R}^{l \times m}$ and $\Sigma \in \mathbb{R}^{l \times l}$. Here, $l \leq n$ is the number of directly actuated state dimensions, so $\R^{l \times l}$ is the directly actuated subspace that noise may act in. Although in some other work \cite{theodorou2010generalized}, $\Sigma$ and $B$ may be functions of $x$ and $t$, Assumption \ref{assume:constant_sigmab} is necessary for the remainder of this paper.

Finally, define the cost function with quadratic control cost and arbitrary state and terminal costs. Let $\phi(x)$ be terminal cost, $q(x, t)$ be state cost matrix, and $R \in \mathbb{R}^{m \times m}$ be the quadratic control cost. Therefore, the expected cost becomes
\begin{multline}\label{cost_def}
    V(x, t) = \min_{u} \, \mathbb{E} \bigg[ \phi(x(t_f)) \\
    + \int_{t}^{t_f} \Big[q(x(\tau), \tau) 
    + \frac{1}{2} u(\tau)^T R \, u(\tau) \Big] d\tau \bigg]
\end{multline}
where the expectation is taken over trajectories of \eqref{system_def} under the admissible control policy $u$.

For a system of the form of \eqref{system_def} and cost function of form \eqref{cost_def}, the Stochastic Hamilton-Jacobi-Bellman equation is
\begin{multline}\label{original_hjb}
    -\partial_t V = q(x, t) + f(x, t)^T \nabla_x V + \frac{1}{2} \mathrm{Tr}(\Sigma_{(f)} \nabla_{xx}V) \\
    - \frac{1}{2} \nabla_x V^T B_{(f)} R^{-1} B_{(f)}^T \nabla_x V 
\end{multline}
where optimal control is $u^* = -R^{-1} B^T \nabla_x V$ \cite{kappen2005path}. However, solving this non-linear PDE is extremely difficult for complex systems. Instead, we introduce the logarithm transform $V=-\lambda \log \Psi$ and make the simplifying assumption that $\Sigma_{(f)} = \lambda B_{(f)} R^{-1} B_{(f)}^T$, or equivalently, $\Sigma = \lambda B R^{-1} B^T$. Under the transform and assumption, the non-linear terms of the stochastic HJB equation cancel, reducing to a linear PDE with condition $\Psi(x(t_f), t_f) = \exp (\lambda^{-1} \phi)$:
\begin{equation}\label{transformed_hjb}
    \partial_t \Psi = \frac{\Psi}{\lambda} q - f^T \nabla_x \Psi 
    - \frac{1}{2} \mathrm{Tr}(\Sigma_{(f)} \nabla_{xx} \Psi)
\end{equation}

This linear PDE belongs to a class that may be transformed into an expectation by the Feynman-Kac Lemma \cite{oksendal2003stochastic}. Under Feynman-Kac, \eqref{transformed_hjb} becomes an expectation of a function of total path cost under uncontrolled dynamics:
\begin{equation}\label{fk_transformed}
    \Psi(x, t) = \mathbb{E}^P \bigg[ \exp \bigg(-\frac{S}{\lambda} \bigg) \bigg]
\end{equation}
where $S = \phi(x(t_f)) + \int_{t}^{t_f}q(x(\tau), \tau)\,d\tau$ is the total uncontrolled path cost and $P$ is the probability measure for uncontrolled dynamics ($u =0$). This expectation may be effectively calculated with Monte Carlo methods by sampling and evaluating uncontrolled trajectories.

\section{Problem Approximation}\label{sec:problem_approximation}
\noindent This section introduces the concept of \textit{problem approximation}: the process of choosing and solving an alternative, tractable problem that provides a bound on the original problem's cost function as well as a valid control sequence. For path integral control, this practice is already widespread, albeit implicit. When applying \eqref{fk_transformed} in practice, the resulting $\Psi$ is calculated for the control cost matrix implied by $\Sigma$ through the linear constraint, rather than the original $R$. As we will demonstrate, the expected cost will be at least the original expected cost because new problem defined by the implied control cost is \textit{harder}.

\begin{definition}[Hardness]\label{def:harder}
    One problem is \textit{harder} than another if its optimal expected cost is greater-than-or-equal to the other's for all possible $x$ and $t$:
    \begin{equation}\label{harder_eq}
        V(x, t | \theta_{easy}) \leq V(x, t | \theta_{hard}) \qquad \forall x \in \R^n, \quad t \in [t_0, t_f]
    \end{equation}
    where $\theta_{easy}$ and $\theta_{hard}$ are the parameters defining the easy and hard problems, respectively.
\end{definition}

Although simply allowing $R$ to be implied by $\Sigma$ is the simplest method, it is neither the only nor generally optimal approximation. The constraint 
$\Sigma = \lambda B R^{-1} B^T$ admits infinitely many solutions if we permit positive semi-definite increases to both $\Sigma$ and $R$. We refer to each such modified problem as a \textit{task}. For any $\Sigma^{+}, R^{+} \geq \mathbf{0}$, a new task is defined by the coupling applied to the augmented parameters.
\begin{equation}\label{loose_constraint}
    \mathbf{0} \leq \Sigma + \Sigma^+ = \lambda B (R + R^+)^{-1}B^T
\end{equation}
Furthermore, we will demonstrate that under the following assumption, this task is harder than the original problem.

\begin{assumption}\label{assume:b_columns}
    The columns of $B$ span $\R^{l}$.
\end{assumption}

\begin{lem}\label{lem: posdef_r_s}
Under Assumption \ref{assume:b_columns}, the cost-to-go of a task is greater than or equal to the original's if its $R$ and $\Sigma$ matrices are positive semi-definitely greater.
 \begin{equation}
 \begin{split}
      V(x, t | \Sigma, R, \lambda ) \leq V(x, t | \Sigma + \Sigma^{+}, R + R^{+}, \lambda)
     \\ \forall \Sigma^{+} \geq \mathbf{0}, R^{+} \geq \mathbf{0}
 \end{split}
 \end{equation}
\end{lem}
\begin{proof}
    This property is true by definition for $R^+$, as the control cost at each timestep will be at least that of the original for any possible control sequence.
    \begin{equation}
        \begin{split}
            u^{T} R u \leq u^{T} (R + R^{+}) u = u^T R u + u^T R^{+} u, \\ \forall u \in \mathbb{R}^{m}, R^{+} \geq \mathbf{0}
        \end{split}
    \end{equation}
    
    As controls are linear with dynamics and quadratic with cost, there exists some optimal control sequence $\widetilde{u}$ \cite{zhou1999stochastic}. Since only directly actuated state dimensions may be affected by noise, per \eqref{bs_partition}, any additional noise $\Sigma^+ \geq \mathbf{0}$ may be considered a stochastic addition to the state-effect of the existing policy.
    \begin{equation}
        B u dt = B\widetilde{u}dt + d\widetilde{\xi}, \quad d\widetilde{\xi} \sim \mathcal{N}(0, \Sigma^+)
    \end{equation}
    Any deviation from the deterministic control sequence within the range of $B$ results in a different but still admissible control policy. Since $\widetilde{u}$ is an optimal solution, no deviation within the admissibly policy set can improve expected cost. From Assumption \ref{assume:b_columns}, $d \widetilde{\xi} \in \mathrm{Range}(B) = \R^{l}$ everywhere. Therefore, the expected cost with covariance $\Sigma + \Sigma^+$ will be greater-than-or-equal to the original.
\end{proof}

Since the solution for $R^+$ is unique for a given $\lambda$ and $\Sigma^+$ under Assumption \ref{assume:b_columns}, $\lambda$ and $\Sigma^+$ can entirely parameterize the space of valid tasks. However, $\Sigma^+$ is now bounded to values that admit a solution for $R^+$. Substituting \eqref{loose_constraint} into inequality $B(R + R^+)^{-1}B^T \leq BR^{-1}B^T$ and rearranging, we obtain:
\begin{equation}\label{sigma_bounds}
    \mathbf{0} \leq \Sigma^+ \leq \lambda BR^{-1}B^T - \Sigma
\end{equation}

\begin{lem}\label{lem: min_lambda}
    The optimal $\lambda$, denoted $\lambda^*$, is always the minimum value admitting solutions to the following linear constraint.
    \begin{equation}\label{lambda_min_eq}
        \min_{\lambda} \quad \mathrm{s.t} \quad \Sigma \leq \lambda B R^{-1} B^{T} 
    \end{equation}
    This value may be calculated as a max eigenvalue problem:
    \begin{equation}
        \lambda^* = \Lambda_{max}((B R^{-1} B^{T})^{-\frac{1}{2}} \Sigma (B R^{-1} B^{T})^{-\frac{1}{2}})
    \end{equation}
\end{lem}
\begin{proof}
For a given valid task defined by $\hat{\Sigma}$, $\hat{R}$, and $\hat{\lambda}$ satisfying $\hat{\Sigma} = \hat{\lambda}B{\hat{R}^{-1}}B^T$, any increase in $\lambda$ can be considered a positive semi-definite addition to covariance.
\begin{equation}
\begin{split}
    \hat{\Sigma} + \hat{\Sigma}^+ = (\hat{\lambda} + \Delta \hat{\lambda}) B{\hat{R}^{-1}}B^T \\
    \hat{\Sigma}^+ = \Delta \hat{\lambda} BR^{-1}B^T
\end{split}
\end{equation}
As demonstrated in Lemma \ref{lem: posdef_r_s}, adding to either variable positive semi-definitely results in a greater-than-or-equal cost-to-go. By monotonicity, the minimum possible $\lambda$ must be optimal.

The minimum possible $\lambda$ is the minimum that admits a solution to \eqref{loose_constraint}. Substituting $\Sigma \leq \Sigma + \Sigma^+$ and $\lambda B (R + R^+)^{-1} B^T \leq \lambda B R^{-1} B^T$ into \eqref{loose_constraint}, we find $\Sigma \leq \lambda B R^{-1} B^T$ as the condition for viability.

According to \cite{boyd_geneig}, LMIs of the form $A \leq \lambda B$ may be solved as the maximum eigenvalue of $B^{-\frac{1}{2}}AB^{-\frac{1}{2}}$.
\end{proof}

Combining Lemma \ref{lem: posdef_r_s} and Lemma \ref{lem: min_lambda}, we find that task choice is parameterized solely by $\Sigma^+$, and the optimal task choice must lie along the following linear---but singular---constraint:
\begin{equation}\label{sigma_bounds_final}
    \mathbf{0} \leq \Sigma^+ \leq \lambda^* B R^{-1} B^T - \Sigma
\end{equation}

Because the expected cost-to-go of any task approximated from a given problem will always be greater than that of the original problem, it serves as an upper bound on the problem's expected cost. Although it is not possible to access the true cost directly with path integral methods, by minimizing expected cost over the space of all valid tasks, we minimize this upper bound on the expected cost-to-go.

\section{Task Choice}\label{sec:task_choice}
\noindent In this section, we propose a method of efficiently optimizing expected cost over the space of valid tasks formulated in Section \ref{sec:problem_approximation}. First, we will demonstrate that it is possible to re-weight existing Monte Carlo samples drawn from any given Gaussian distribution in terms of any other Gaussian probability measure \textit{without resampling} in discrete time. Then, we will formulate the problem of minimizing expected cost as an optimization problem over the covariance of this new probability measure.

\subsection{Covariance Scaling}\label{subsec:covariance_scaling}
\noindent To find the expected cost under an updated covariance, $\Sigma + \Sigma^+$, we will employ the likelihood ratio between the original probability measure and one defined by updated covariance. The approach is similar to the generalized importance sampling derived in \cite{williams2017model}, although the purpose is different.  

\begin{assumption}\label{assume:sigma}
    $\Sigma$ is non-singular.
\end{assumption}

\begin{lem}\label{lem:covariance_scaling}
    Under Assumption \ref{assume:sigma}, the log-transformed expected cost $\Psi$ for a task parameterized by $\Sigma^{+}$ is given by
    \begin{multline}\label{scaled_feynman_kac}
        \Psi = \mathbb{E}^{P} \biggl[ \biggl[ \frac{|\Sigma|}{|\Sigma + \Sigma^+|} \biggr]^{\frac{N}{2}} \\
        \cdot \exp \biggl( \frac{1}{2}
        \mathrm{Tr}((\Sigma^{-1} - (\Sigma + \Sigma^+)^{-1})Z) - \frac{S}{\lambda^*} \biggr) \biggr] \\
        Z = \sum_{t=t_0}^{t_f} z(t)  z(t)^T, \qquad
        z(t)=\frac{\Delta x(t)}{\Delta t} - f(x(t), t)
    \end{multline}
    where $N$ is the number of time steps and $Z$ is the scatter matrix of deviations from deterministic dynamics.
\end{lem}

\begin{proof}
    Define new gaussian probability measure $Q$ with covariance $\Sigma + \Sigma^+$. Assumption \ref{assume:sigma} implies that $Q \sim P$, as both probability measures assign non-zero probability to all values in $\R^{l \times l}$. If $Q \sim P$, then there exists some $f(\omega)$ such that $dQ(\omega) = g(\omega)dP(\omega)$, and likewise some $g_x(x_{[t_0,t]})$ such that $dQ_x(x_{[t_0,t]}) = g_x(x_{[t_0,t]})dP_x(x_{[t_0,t]})$. Applying this change of measure to \eqref{fk_transformed}, we find:
    
    \begin{equation}\label{rebased_feyman_kac}
        \mathbb{E}^{Q} \biggl[ \exp \biggl( - \frac{S}{\lambda} \biggr) \biggr] = 
        \mathbb{E}^{P} \biggl[ g_{x}(x_{[t_0, t_f]}) \cdot \exp \biggl( - \frac{S}{\lambda} \biggr) \biggr]
    \end{equation}
    By the Markov property, the likelihood ratio is the product of the probability of the next state, given the current state, and the probability of its history. Then, from the system definition in \eqref{system_def}, the likelihood of the next state given the current is a function of the state derivative.
    \begin{equation*}
    \begin{split}
        \frac{dQ_x}{dP_x} (x_{[t_0,t + \Delta t]})
        = \frac{dQ_x}{dP_x} (x_{t + \Delta t} | x_{t}) \cdot \frac{dQ_x}{dP_x} (x_{[t_0,t]}) \\
        = \bigg[ \frac{dQ_x}{dP_x} \bigg( \frac{\Delta x_{t + \Delta t}}{\Delta t}  
        \bigg) \bigg] \cdot \frac{dQ_x}{dP_x} (x_{[t_0,t]})
    \end{split}
    \end{equation*}
    Combining the recurrence relation, we obtain:
    \begin{equation}
        g_{x}(x_{[t_0, t_f]}) = \prod_{i=0}^{N} \bigg[
            \frac{dQ_x}{dP_x} \bigg( \frac{\Delta x_{t_i + \Delta t}}{\Delta t}  
        \bigg)
        \bigg]
    \end{equation}
    Substituting the Gaussian PDF with mean $f(x, t)$ for both measures yields
    \begin{equation}\label{likelihood_ratio_final}
    \begin{split}
        g_{x} = \biggl[ \frac{| \Sigma_{P} |}{| \Sigma_{Q} |} \biggr]^{\frac{N}{2}} \cdot 
        \exp \biggl( \frac{1}{2}
        \mathrm{Tr}((\Sigma_{P}^{-1} - \Sigma_{Q}^{-1})Z) 
        \biggr)
    \end{split}
    \end{equation}
    where multiplied exponentials have been contracted into a summation inside the exponential function and path-invariant terms have been either canceled or been combined into the determinant ratio.
    Finally, substituting \eqref{likelihood_ratio_final} and $\Sigma_Q$ and $\Sigma_P$ definitions into \eqref{rebased_feyman_kac}, we obtain \eqref{scaled_feynman_kac}.
\end{proof}

\begin{remark}\label{remark:discrete}
    The likelihood ratio method as applied in this section is only possible in discrete time. In continuous time, likelihood ratios between measures of different covariance are always zero or positive infinity. This can be seen by taking the limit of \eqref{likelihood_ratio_final} as $\Delta t \rightarrow 0$.
\end{remark}

\subsection{Task Optimization}
\noindent Lemma \ref{lem:covariance_scaling} provides an expression for $\Psi$ in terms of additive covariance $\Sigma^+$. As $\Sigma^+$ parameterizes the entire space of valid tasks, maximizing \eqref{scaled_feynman_kac} gives the optimal task, providing the lowest upper bound on the original problem's expected cost. In its current form, the matrix decision variable is inverted. However, we show that with a simple change of variables, \eqref{scaled_feynman_kac} becomes a practical optimization problem over the space of all valid tasks.

\begin{prop}\label{prop:opt_problem}
    After the invertible change of variables
    \begin{equation}\label{D_def}
        (\Sigma + \Sigma^+)^{-1} = \Sigma^{-1} - D
    \end{equation}
    the optimal task choice may be found as the solution to the following optimization problem.
    \begin{equation}\label{D_opt_final}
        \max_{D} \Bigg\{ 
            {(|\Sigma||\Sigma^{-1} - D|})^{\frac{N}{2}} 
            \sum_{k}^{K}
            \exp \bigg( \frac{1}{2}
                \mathrm{Tr}(DZ_k) - \frac{S_k}{\lambda^*}
            \bigg)
        \Bigg\}
    \end{equation}
    \begin{equation}\label{D_bounds}
    0 \leq D \leq \Sigma^{-1} - \lambda^{*^{-1}} (B R^{-1} B^{T})^{-1}
    \end{equation}
\end{prop}

\begin{proof}
    Substituting \eqref{D_def} into \eqref{scaled_feynman_kac}, simplifying, and converting the expectation to an empirical average over K Monte Carlo paths gives
    \begin{multline}\label{D_opt_real}
        \max_{D} \Bigg\{ 
            {\bigg[ \frac{|\Sigma|}{|(\Sigma^{-1} - D)^{-1}|}} \bigg]^{\frac{N}{2}} \cdot \\
            \sum_{k}^{K}
            \exp \bigg( \frac{1}{2}
                \mathrm{Tr}(DZ_k) - \frac{S_k}{\lambda^*}
            \bigg)
        \Bigg\}
    \end{multline}
    where $Z_k$ and $S_k$ are the scatter matrix and total cost terms for the $k$-th sample, respectively. Recognizing that the determinant of an inverse is the inverse of the determinant and canceling terms simplifies the expression to \eqref{D_opt_final}.
    
    Since $A \leq B \rightarrow B^{-1} \leq A^{-1}, \quad \forall A, B > \mathbf{0}$, then \eqref{sigma_bounds_final} can be inverted:
    \begin{equation*}
        \lambda^{*^{-1}} (B R^{-1} B^{T})^{-1} \leq (\Sigma - \Sigma^{+})^{-1} \leq \Sigma^{-1}
    \end{equation*}
    After substituting in \eqref{D_def} and rearranging, we obtain \eqref{D_bounds}.

    $\Sigma^+$ subject to \eqref{sigma_bounds_final} parameterizes the entire space of valid and potentially optimal tasks. Since $\Sigma^+ \rightarrow D$ is a bijective mapping, $D$ subject to \eqref{D_bounds} also parameterizes the entire task space. Therefore, the global maximizer of \eqref{D_opt_final} is the globally optimal task.
\end{proof}

\subsection{Numerical Considerations}\label{subsec:numerical_considerations}
\noindent Since the upper bound on $D$ given in \eqref{D_bounds} will always be singular, solving the problem in \eqref{D_opt_final} \textit{as is} would be quite difficult. Instead, the problem must be reformulated in the $D$ upper bound's non-singular subspace. 

This non-singular subspace may be found with eigen-decomposition of the upper bound of $D$:
\begin{equation}
    \Sigma^{-1} - {\lambda^*}^{-1} B R^{-1} B^{T} = E M E^T
\end{equation}
where eigenvalues are sorted in ascending order and at least one eigenvalue will be $0$ due to the choice of $\lambda^*$. Therefore, under the change of basis $E$, we may partition $D$ and $\Sigma^{-1} - \lambda^* B R^{-1} B^{T}$ along the same indices.
\begin{equation}
    D = E\begin{bmatrix} 
        0 & 0 \\ 
        0 & \widetilde{D} 
    \end{bmatrix}E^T, 
    \qquad
    EME^T = E \begin{bmatrix}
        0 & 0 \\
        0 & \widetilde{M}
    \end{bmatrix} E^T
\end{equation}
Determinants and quadratic forms are independent of basis, so \eqref{D_opt_final} can be represented in the basis of $E$ as
\begin{multline}
    \max_{\widetilde{D}} \Bigg\{ 
        {(|\Sigma||E^{T} (\Sigma^{-1} - D) E|})^{\frac{N}{2}} \cdot \\
        \sum_{k}^{K}
        \exp \bigg( \frac{1}{2}
             \mathrm{Tr}(\widetilde{Z}_k \widetilde{D}) - \frac{S_k}{\lambda^*}
        \bigg)
    \Bigg\}
\end{multline}
where $\widetilde{Z} = \big[ E^{T} Z E \big]_{(p)}$ and the $(p)$ index denotes the bottom-right sub-matrix associated with positive eigenvalues in $M$ and excluding the dimensions where $E^T D E$ is zero.

Finally, due to the block structure of $D$, we can apply a Schur Complement to the determinant to reduce the dimensionality of the determinant. First, partition $E^T \Sigma^{-1} E$ along the same dimensions as $D$.
\begin{equation}
    T = E^{T} \Sigma^{-1} E = \begin{bmatrix}
        T_{(z)} & T_{(c)}^T \\
        T_{(c)} & T_{(p)}
    \end{bmatrix}
\end{equation}
Then, perform the Schur Complement on the matrix determinant and consolidate constants into $H$.
\begin{equation}\label{D_opt_reduced}
    \begin{split}
    \max_{\widetilde{D}} \Bigg\{ 
        {(|\Sigma||T_{(z)}||H - \widetilde{D}|})^{\frac{N}{2}} 
        \sum_{k}^{K}
        \exp \bigg( \frac{1}{2}
            \mathrm{Tr}(\widetilde{Z} \widetilde{D}) - \frac{S_k}{\lambda^*}
        \bigg)
    \Bigg\} \\
    \mathbf{0} \leq \widetilde{D} \leq \widetilde{M}, \qquad
    H = T_{(p)} - T_{(c)} T_{(z)}^{-1} T_{(c)}^T
    \end{split}
\end{equation}

With this change of variables, we have an optimization problem over the full-rank PSD decision variable $\widetilde{D}$.

\subsection{Wishart Importance Sampling}\label{subsec:wishart}
\noindent In previous sections, we found an optimization problem theoretically yielding the optimal task, yet practical issues remain.
Each sample's contribution to $\Psi(x, t)$ will decay rapidly where $Z_k$ is distant from $(\Sigma^{-1} - D)^{-1}$. However, $Z_k$ will tend towards the sampling distribution's covariance, so the estimate given in \eqref{D_opt_final} will only be accurate in that neighborhood. Accuracy across the entire optimization domain would require drawing samples from a continuous range of covariances. Although such a scheme complicates importance sampling, the Normal-Inverse-Wishart (NIW) prior provides a closed-form solution for the marginal likelihood of $Z$. The NIW models the process of drawing a covariances matrix from an Inverse-Wishart distribution, then drawing $N$ samples from the Gaussian defined by the matrix. However, $\Sigma^+$ is under the singular constraint of \eqref{sigma_bounds}, and most samples drawn from the Inverse-Wishart distribution will not lie along this constraint, contributing little to the cost estimate. To address this shortfall, we define the Inverse-Wishart prior over the non-singular subspace of \eqref{sigma_bounds}:
\begin{equation}
    \Sigma^+ = G\begin{bmatrix}
        0 & 0 \\ 0 & \Xi_+
    \end{bmatrix}G^T, \qquad \lambda^* B R^{-1} B^T - \Sigma = G V G^T
\end{equation}
where $\Xi_+ \in \R^{(l - 1) \times (l - 1)}$ is PSD and $G$ is the eigenvector matrix of $\lambda^* B R^{-1} B^T - \Sigma$. Then, adding $\Sigma$ to get the full covariance matrix, we transform $\Sigma$ into the basis $G$ and factor it into a singular matrix $[\alpha, \beta][\alpha, \beta]^T$ and the partitioned block sub-matrix, $\Xi_\Sigma$.
\begin{equation}
    \Sigma + \Sigma^+ = G \bigg(\begin{bmatrix}
        \alpha^2 & \alpha \beta^T \\ \alpha \beta & \beta \beta^T
    \end{bmatrix} + \begin{bmatrix}
        0 & 0 \\ 0 & \Xi_\Sigma + \Xi_+
    \end{bmatrix} \bigg) G^T
\end{equation}
Since we need to sample a single PSD matrix from the prior, absorb $\Xi_+$ and $\Xi_\Sigma$ into $\Xi$. Furthermore, we can transform the resulting $Z$ draw into the basis of $G$ to simplify the Wishart definition. We now have a well defined NIW system:
\begin{equation}
\begin{split}
    \Xi \sim \mathcal{IW}_{(l-1) \times (l-1)}(v, Y), \qquad Z_G \sim \mathcal{W}_{l\times l} (U, N) \\
    U = \begin{bmatrix}
        \alpha^2 & \alpha \beta^T \\ \alpha \beta & \beta \beta^T + \Xi
    \end{bmatrix}, \qquad Z_G = G^T Z G
\end{split}
\end{equation}

\begin{lem}
\label{lem:constrained_marginal}
    Under the constrained Inverse-Wishart prior defined over the non-singular subspace of \eqref{sigma_bounds}, the marginal likelihood of $Z$ is
    \begin{equation}\label{marg_final}
    \begin{split}
        p(Z) = \frac{\exp\big( -\frac{Z_G^{(z)}}{2\alpha^2} \big) \Gamma_{l-1}(\frac{N + v}{2}) |Y|^{v/2} |Z|^{(N - l - 1)/2}}
        {2^{N/2}  \alpha^N \Gamma_l(\frac{N}{2}) \Gamma_{l-1}(\frac{v}{2}) |\bar{Z}_G + Y|^{(N + v)/2}} \\
        \bar{Z}_G := \alpha^{-2} Z_G^{(z)}\beta \beta^T + Z_G^{(p)} - 2 \alpha^{-1}Z_G^{(c)} \beta^T
    \end{split}
    \end{equation}
\end{lem}

\begin{proof}
The marginalization over $Z$ is an integral over the domain of the likelihood:
\begin{equation}
    p(Z) = \int_{\mathcal{S}^+} p(Z_G | \Xi) p(\Xi) d\Xi
\end{equation}
where $\mathcal{S}^+$ is the set of PSD matrices. Next, we substitute in the prior and likelihood expressions.
\begin{multline}\label{original_marg_integral}
    p(Z) = \int_{\mathcal{S^+}} 
        \frac{|Z_G|^{(N - l - 1)/2} |Y|^{v/2}}
            {2^{(Nl + vl - v)/2} |U|^{N/2} \Gamma_l (\frac{N}{2}) \Gamma_{l - 1}(\frac{v}{2}) |\Xi|^{(v + l)/2} } \\
    \cdot \exp (-\frac{1}{2}( \mathrm{Tr}(Z_G U^{-1}) + \mathrm{Tr}(Y \Xi^{-1})) d\Xi
\end{multline}

Now, since the likelihood is defined by $U$ while the prior is over $\Xi$, we do not automatically have conjugacy as demonstrated in \cite{murphy2007conjugate}. The problematic terms are $\mathrm{Tr}(Z_G U^{-1})$ and $|U|$. These terms can be shown to be proportional to $\mathrm{Tr}(Z_G \Xi^{-1})$ and $\Xi$, so conjugacy may be regained. We will first demonstrate this for $|U|$ by applying a Schur complement along the existing partition:
\begin{equation}\label{u_det}
    |U| = \alpha^2 |\beta \beta^T + \Xi - \alpha^{-2} \alpha \beta \alpha \beta^T | 
    = \alpha^2 |\Xi|
\end{equation}
Next, simplifying $\mathrm{Tr}(Z_G U^{-1})$ requires a block inversion along the existing partition. After minor simplification, we find:
\begin{equation}
    U^{-1} = \begin{bmatrix} 
        \alpha^{-2} + \alpha^{-2} \beta^T \Xi^{-1} \beta & -\alpha^{-1} \beta^T \Xi^{-1} \\
        -\alpha^{-1} \Xi^{-1} \beta & \Xi^{-1}
    \end{bmatrix}
\end{equation}
$\mathrm{Tr}(Z_G U^{-1})$ may alternatively be considered a sum over all indices of the Hadamard product $Z_G \odot U^{-1}$. Additionally, $\beta^T \Xi^{-1} \beta = \mathrm{Tr}(\beta\beta^T \Xi^{-1})$, so the trace becomes
\begin{multline}
    \mathrm{Tr}(Z_G U^{-1}) = \\
    \alpha^{-2} Z_G^{(z)} + \mathrm{Tr}(\alpha^{-2} Z_G^{(z)}\beta \beta^T \Xi^{-1}) + \mathrm{Tr}(Z_G^{(p)}\Xi^{-1}) -  \\\alpha^{-1} \mathbf{1}^T (Z_G^{(c)} \odot \beta^T \Xi^{-1})\mathbf{1}
        - \alpha^{-1} \mathbf{1}^T (Z_G^{(c)} \odot \Xi^{-1} \beta)\mathbf{1}
\end{multline}
where $Z_G$ has been partitioned along the same axes as $U^{-1}$:
\begin{equation}
    Z_G = \begin{bmatrix}
        Z_G^{(z)} & Z_G^{(c)} \\ Z_G^{(c)} & Z_G^{(p)}
    \end{bmatrix}
\end{equation}
Finally, a sum over elements of the Hadamard product of two vectors is equivalent to their dot product. Therefore, $\mathbf{1}^T (Z_G^{(c)} \odot \beta^T \Xi^{-1})\mathbf{1} = {\beta}^T \Xi^{-1} Z_G^{(c)} = \mathrm{Tr}(Z_G^{(c)} \beta^T \Xi^{-1})$. Since $\Xi^{-1}$ is symmetric, $\mathbf{1}^T (Z_G^{(c)} \odot \Xi^{-1}\beta)\mathbf{1} = \mathrm{Tr}(Z_G^{(c)} \beta^T \Xi^{-1})$ as well. Applying these simplifications as well as the linearity of traces, we get
\begin{equation}\label{z_bar_def}
\begin{split}
    \mathrm{Tr}(Z_G U^{-1}) = \alpha^{-2} Z_G^{(z)} + \mathrm{Tr}(\bar{Z}_G \Xi^{-1})
\end{split}
\end{equation}
where $\bar{Z}_G$ is as defined in \eqref{marg_final}.

Returning to the marginalization, substitute \eqref{u_det} and \eqref{z_bar_def} into \eqref{original_marg_integral} and perform basic simplifications.
\begin{multline}\label{subbed_marg_integral}
    p(Z) = \int_{\mathcal{S^+}} 
        \frac{|Z_G|^{(N - l - 1)/2} |Y|^{v/2}}
            {2^{(Nl + vl - v)/2} \Gamma_l (\frac{N}{2}) \Gamma_{l - 1}(\frac{v}{2}) \alpha^N  |\Xi|^{(N + v + l)/2}} \\
            \cdot \exp \bigg(-\frac{Z_G^{(z)}}{2\alpha^2} \bigg) \cdot \exp \bigg( -\frac{1}{2}\mathrm{Tr}((Y + \bar{Z}_G) \Xi^{-1} \bigg) d\Xi
\end{multline}
Now, recognize the Inverse-Wishart kernel of $\mathcal{IW}_{(l-1) \times (l-1)}(N + v, Y + \bar{Z}_G)$, which must integrate to one. Therefore, marginal likelihood must be the ratio between the existing terms constant \textit{wrt} $\Xi$ and the normalization terms of the updated Inverse-Wishart kernel. This ratio simplifies to \eqref{marg_final}.
\end{proof}

Applying this marginal likelihood expression to importance sampling, the likelihood ratio actually simplifies if we use the Wishart likelihood for the variable probability measure parameterized by $D$.
\begin{multline}\label{wishart_d_likelihood}
    p(Z | D) = \frac
        {|\Sigma^{-1} - D|^{N/2}|Z|^{(N - l - 1)/2}}
        {2^{Nl/2}  \Gamma_l (\frac{N}{2})} \\
    \cdot \exp \Big( \frac{1}{2} \mathrm{Tr}(Z(D-\Sigma^{-1})) \Big)
\end{multline}

\begin{prop}\label{prop:wishart_cost}
    Under the constrained NIW sampling scheme with marginal likelihood given by Lemma \ref{lem:constrained_marginal}, the desirability function is
    \begin{multline}\label{wishart_feynman_kac}
        \Psi = 
        \frac
            {\alpha^N \Gamma_{l-1}(\frac{v}{2}) \big) }
            { 2^{N(l - 1)/2} \Gamma_{l-1}(\frac{N + v}{2}) |Y|^{v/2} } |\Sigma^{-1} - D|^{\frac{N}{2}} \sum_{k = 1}^K \\ 
         \bigg[
            |\bar{Z}_{G_k} + Y|^{\frac{N + v}{2}} 
            \exp \bigg( \frac{1}{2} 
                \mathrm{Tr}(Z_k(D - \Sigma^{-1})) -\frac{S_k}{\lambda^*} + \frac{{Z_{G_k}^{(z)}}}{2\alpha^2} 
            \bigg)
        \bigg]
    \end{multline}
\end{prop}

\begin{proof}
    Substituting the ratio between \eqref{wishart_d_likelihood} and \eqref{marg_final} into \eqref{rebased_feyman_kac}, canceling terms, and converting to an empirical average over $K$ samples yields \eqref{wishart_feynman_kac}.
\end{proof}

Although \eqref{wishart_feynman_kac} is quite complicated, all additional terms except $|\bar{Z}_{G_k} + Y|$ and $Z_{G_k}^{(z)}$ are sample-invariant. All $D$-dependent terms are identical in form to \eqref{D_opt_final}, so the reduced-dimension formulation of Section \ref{subsec:numerical_considerations} applies directly.

\begin{remark}\label{recursive_wishart}
    The methods described in this section may be extended to integrate over an Inverse-Wishart prior of lower rank as well by recursively applying \eqref{u_det} and \eqref{z_bar_def} with successive bases and rank one perturbations.
\end{remark}
\subsection{Simple Algorithm}\label{subsec:algorithm}
\noindent In this section, we present a very simple, online algorithmic implementation of the BOTC framework. The most obvious approach is to  insert the task choice step between trajectory simulations and control generation, as defined in the original \cite{kappen2005path} formulation. 

\begin{algorithm}
    \caption{Simple Online BOTC-PI}
    \KwData{
        $R$, $\Sigma$, $\phi$, $q$: problem parameters}

    \While{task not complete}{
        Collect $K$ trajectory samples \;
        Solve optimization problem for $D$ from either \eqref{D_opt_final}, \eqref{D_opt_reduced}, or \eqref{wishart_feynman_kac} \;
        $u^* \leftarrow \sum_{i=0}^{K} g_x(D) \exp \Big( -\frac{S_k}{\lambda^*} \Big) \frac{d\xi_i}{dt}$ \;
        Send $u^*(0)$ to actuators \;
    }
\end{algorithm}

\section{Simulation Results}\label{sec:simulations}

\begin{figure}
    \centering
    \includegraphics[width=1.0\linewidth]{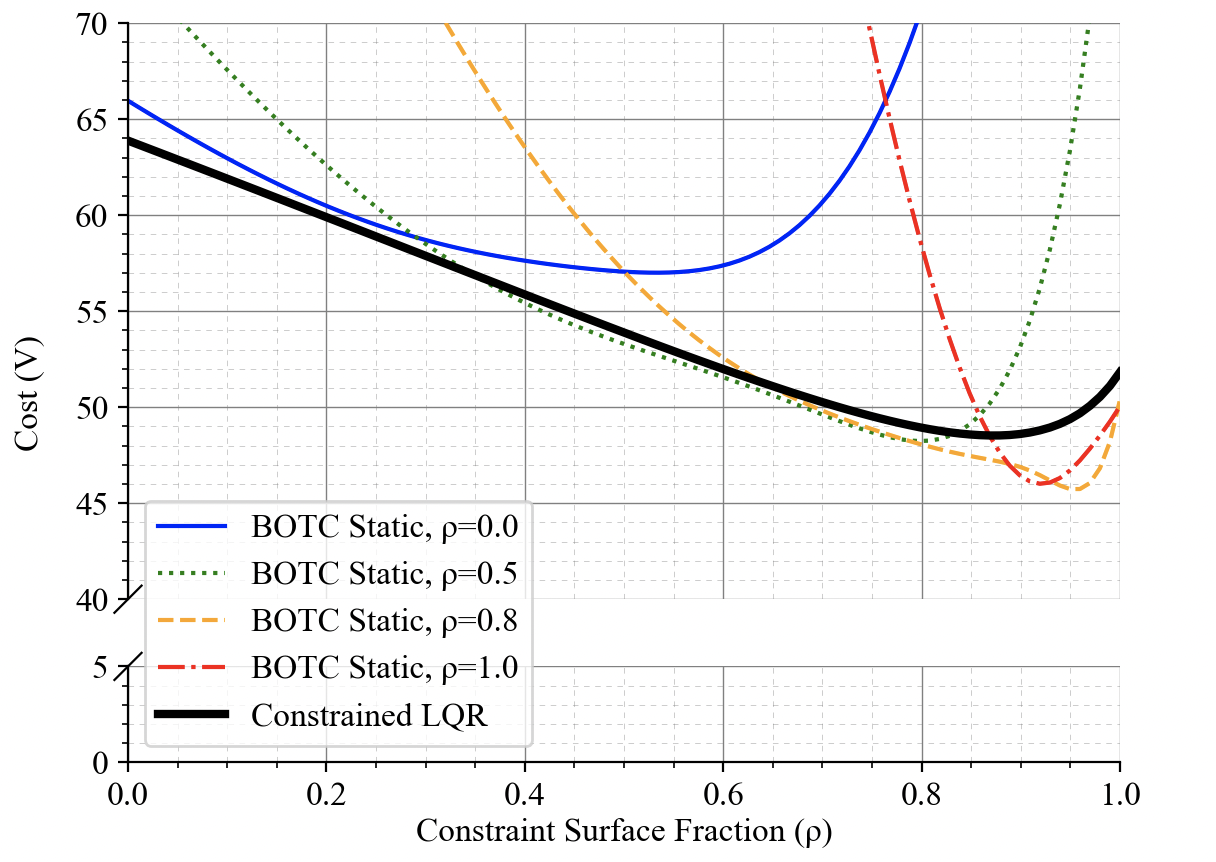}
    \caption{Static sampling: estimated (colored) and true (black) expected task cost as a function of constraint surface fraction $\rho$, with $K=1000$ and $\Delta t = 0.05$.}
    \label{fig:static}
    \includegraphics[width=1\linewidth]{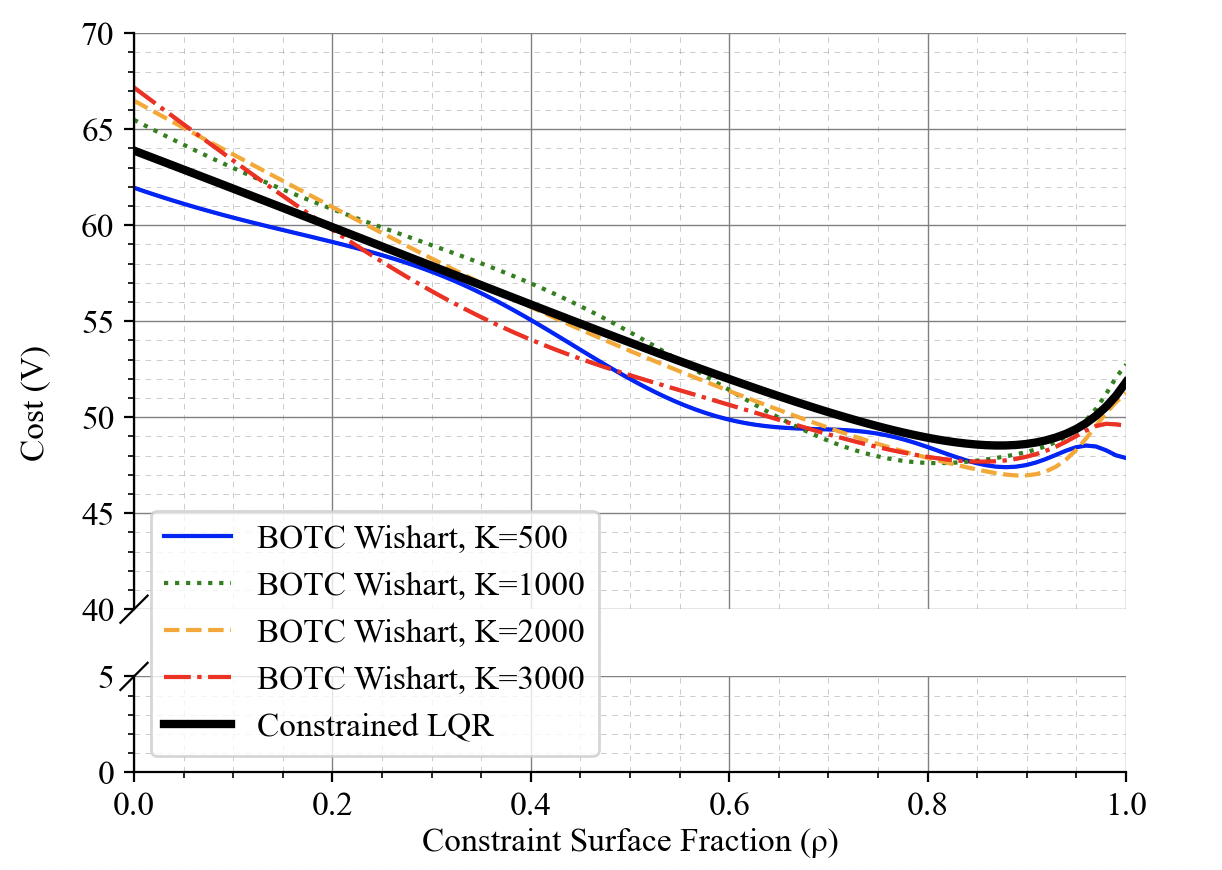}
    \caption{Wishart sampling: estimated expected task cost across the full constraint surface, with Inverse-Wishart mean at $\rho=0.5$ and $v=10$.}
    \label{fig:wishart}
\end{figure}

\noindent In this section, we present the results of simulated experiments on a simple stochastic, two-dimensional, finite-horizon LQR system. We aim to demonstrate that the BOTC framework outlined in Section \ref{sec:task_choice} accurately estimates the expected \textit{task cost}---the cost under the modified control cost and covariance---across the entire constraint surface. We chose LQR for this experiment because the expected cost may be calculated analytically for any given initial position and provides a ground truth for expected task cost across the entire constraint surface. 

The system used for this experiment is defined by:
\begin{equation*}
\begin{split}
    A = \begin{bmatrix}
        0.1 & 0.1 \\ 0 & 0.1
    \end{bmatrix},
    \qquad
    B = \begin{bmatrix}
        1 & 0 \\ 1 & 1
    \end{bmatrix} \\
    \Sigma, R = \begin{bmatrix}
        1 & 0 \\ 0 & 2
    \end{bmatrix},
    \qquad
    Q = \begin{bmatrix}
        1 & 0 \\ 0 & 1
    \end{bmatrix}
\end{split}
\end{equation*}
where $A$ is uncontrolled dynamics, $B$ is the control effect matrix, $Q$ is both running and terminal state costs, and $\Sigma$ and $R$ are as previously defined. Note that this system does not conform to the PI constraint \eqref{pi_constraint}. The infinite-horizon LQR optimal feedback gain is used as the importance sampling mean for all distributions. Initial state for all tests is $x(0) = [3 \quad4]^T$. Although the shapes of the resulting curves are entirely state-dependent, performance is analogous for all initial states. The initial state was chosen to give an interior minimizer for demonstration purposes.

Because the problem has two control dimensions, the optimization problem is one-dimensional as described in Section \ref{subsec:numerical_considerations}. For static sampling experiments, sampling is performed from various distributions corresponding to points along the $\mathbf{0} \leq \widetilde{D} \leq \widetilde{M}$ optimization domain. For readability, sampling distributions are parameterized by fraction along constraint surface, $\rho \in [0, 1]$. All static sampling distributions were calculated with $K=1000$ samples and $\Delta t = 0.05$. Wishart sampling was performed with an inverse-Wishart mean midway along the constraint surface and $v=10$.

As demonstrated in Figure \ref{fig:static}, the static sampling method provides an estimate closely tangent to the true cost surface at the point corresponding to its sampling distribution. As such, in the neighborhood of its sampling distribution, it provides a highly accurate estimate of both the expected cost and its gradient \textit{wrt} the decision variable. 

Figure \ref{fig:wishart} illustrates both the advantages and drawbacks of Wishart sampling. Wishart sampling provides a reasonable estimate of the expected cost across the entire optimization domain. However, its estimate of the curve is occasionally erratic, and therefore its derivative is not as reliable. As demonstrated by the improvement from the $K=500$ curve to the higher sample count curves, Wishart cost surface estimates become significantly more accurate with greater samples.

The contrast between Figures \ref{fig:static} and \ref{fig:wishart} suggests that static sampling is preferable for local optimization, while Wishart sampling is required for global optimization.

\section{Conclusion}

\noindent In this paper, we have presented Bound-Optimized Task Choice (BOTC), a framework for stochastic control extending path integral methods to problems not satisfying the coupling constraint on control cost and covariance. This framework provides methods for finding the globally optimal PI-compliant approximate problem, thereby minimizing the upper bound on expected cost.

To achieve this, we have derived the theoretical basis for mapping an optimal control problem to possible approximations on that constraint's surface. We found the conditions under which an approximated task gives an upper bound on the original problem's expected cost. Then, we reduced that space of valid tasks to the subset of potentially optimal tasks. Next, we adapted existing methods to the purpose of calculating expected cost under an arbitrary probability measure without resampling. Optimizing this expected cost over the space of viable tasks allows us to find the globally optimal task, providing the lowest upper bound on the original problem's expected cost attainable with PI control. We present a novel importance sampling scheme based on the Normal-Inverse-Wishart distribution to improve global task space optimization. Finally, we validate these methods on a simple LQR system, demonstrating both the potential and limitations of this framework.

Further work will apply BOTC methods to concrete non-linear control problems, including both offline problem analysis and online task optimization. Future iterations of the basic algorithm presented in Section \ref{subsec:algorithm} will apply iterative methods for task optimization. The local accuracy of static sampling lends itself to gradient methods, while the global accuracy Wishart sampling will enable such an algorithm to escape local minima.


While this paper uses the most basic PI formulation as proposed by \cite{kappen2005path}, future work should attempt to integrate BOTC methods with other PI frameworks such as MPPI \cite{williams2017model} and the information theoretic formulation in \cite{theodorou2012relative}.

\bibliographystyle{IEEEtran}
\bibliography{bibliography}

\end{document}